\documentclass{article}
\usepackage[utf8]{inputenc}
\usepackage{amssymb}
\usepackage{authblk}
\usepackage{amsfonts}
\usepackage{mathrsfs}
\usepackage{amsmath}
\usepackage{graphicx}
\usepackage{tikz}
\usepackage{dsfont}
\usepackage{float}
\usepackage{diagbox} % 加载宏包
\newcommand{\norm}[1]{\left\lVert#1\right\rVert}
\usepackage[utf8]{inputenc}
\usepackage{graphics}
\usepackage{graphicx}
\usepackage{amssymb, amsthm, amsmath}
\usepackage[all]{xy}
\usepackage[margin=1.4in,marginparwidth=1in,marginparsep=.05in]{geometry}
\PassOptionsToPackage{hyphens}{url}
\usepackage{tikz-cd}
\usepackage{titlesec}
\usepackage{xcolor,graphicx}

\usepackage{verbatim}
\usepackage{float}
\usepackage{mathrsfs}
\usepackage{extarrows}
\usepackage{tikz}
\usetikzlibrary{shapes,arrows,positioning}
\usepackage{mathrsfs}
\usepackage{mathtools}

\usepackage{subfig}
\usepackage{hyperref}
\usepackage{tikz}
\usetikzlibrary{arrows}

\makeatletter
\newcommand*{\rom}[1]{\expandafter\@slowromancap\romannumeral #1@}

\titlespacing*{\section}{0pt}{\baselineskip}{\baselineskip}

\titleformat{\subsection}{\normalfont\bfseries}{\thesubsection.}{1em}{}
\titleformat{\subsubsection}{\normalfont}{\thesubsubsection.}{1em}{\itshape}

\theoremstyle{plain}
\newtheorem{theorem}{Theorem}[section]

\newtheorem{lemma}[theorem]{Lemma}

\theoremstyle{definition}

\theoremstyle{remark}
\newtheorem{remark}[theorem]{Remark}

\numberwithin{equation}{section}

\let\pa\partial

\newcommand{\R}{\mathbb R}
\newcommand{\bA}{\mathbf A}

\newcommand{\bD}{\mathbf D}

\newcommand{\bH}{\mathbf H}
\newcommand{\bI}{\mathbf I}

\newcommand{\bP}{\mathbf P}
\newcommand{\bG}{\mathbf G}

\newcommand{\bb}{\mathbf b}
\newcommand{\bg}{\mathbf g}

\newcommand{\blf}{\mathbf f}
\newcommand{\bn}{\mathbf n}
\newcommand{\be}{\mathbf e}
\newcommand{\bp}{\mathbf p}

\newcommand{\bt}{\mathbf t}

\newcommand{\bu}{\mathbf u}
\newcommand{\bU}{\mathbf U}

\newcommand{\bbf}{\mathbf f}

\newcommand{\divG}{{\mathop{\,\rm div}}_{\Gamma}}
\newcommand{\gradG}{\nabla_{\Gamma}}
\newcommand{\nablaG}{\nabla_{\Gamma}}

\newcommand{\bid}{\mathbf {id}}

\renewcommand{\div}{\textrm{div}\ \!}

\newcommand{\tr}{{\rm tr}}

\DeclareGraphicsExtensions{.pdf,.eps,.ps,.eps.gz,.ps.gz,.eps.Y}

\newcommand*\Laplace{\mathop{}\!\mathbin\bigtriangleup}
\newcommand{\jump}[1]{[#1]}
\title{A note on balance laws of \\coupled surface and bulk fluid flows}

\date{August 1, 2020}
\author[1]{Qi Sun}
\author[2]{Vladimir Yushutin}
\affil[1]{Department of Mathematics, University of Houston}
\affil[2]{Department of Mathematics, University of Maryland}

\begin{document}

\maketitle
\begin{abstract}
Balance laws of coupled bulk-surface-bulk fluid flows are considered. We carefully investigate how do the interface conditions influence the total mass, total momentum and total energy of the system.
\end{abstract}
\section{Introduction}
There is a growing interest in development of numerical methods for flows on thin  fluidic membranes. In this note we aim to present and to examine models of a two-phase bulk flow of Newtonian fluid coupled with a surface flow of Boussinesq-Scriven fluid. Our main concern are the balance laws of mass, momentum and energy which can be easily violated by improper coupling conditions.

Let $\Omega\subset \mathbb{R}^d$ be a given domain, $d=2$ or $d=3$. $\partial\Omega$ at least lipschitz smooth. Consider a time dependent at least $C^2$ smooth closed hyper-surface $(\Gamma(t))_{t\in [0,T]}$ separating $\Omega=\Omega_\pm$ into two domains, $\Omega_{+}(t)$ and $\Omega_{-}(t):=\Omega\setminus \overline{\Omega}_{+}(t) $. We let $\bn$ denote the unit normal on $\Gamma(t)$ pointing towards $\Omega_{+}(t)$, see Fig.\ref{fig:Geom}. For any vector field $\bu$  on $\Gamma$ we define the normal and tangential components, correspondingly: 
\begin{align}\label{comp_G}
u_N:=\bu \cdot \bn\,,\qquad{}\bu_T := \bP \bu = \bu-u_N\bn\end{align}
where $\bP=\bI - \bn^T\bn$ is the projector on the tangent space $T\Gamma(t)$. It is assumed that $(\Gamma(t))_{t\in[0,T]}$ is a $C^2$ smooth hypersurface without boundary that evolves with a time-dependent vector field
${\mathcal{V}(t)}$ and we would like to stress on the pure geometrical nature of this vector field that later will be coupled with physical quantities. 
\begin{figure}[H]\label{Geom}
    \centering
    \includegraphics[scale=0.35]{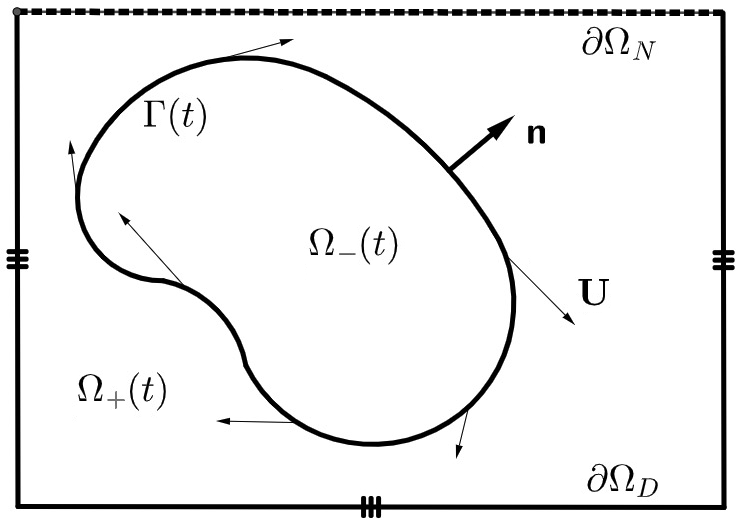}
    \caption{An illustration for the geometrical setup of the problem, $d=2$}
    \label{fig:Geom}
\end{figure}
The structure of the note is as follows. In Section \ref{NS} we briefly discuss   bulk Newtonian flows on $\Omega_\pm(t)$ coupled with two classical interface conditions on $\Gamma(t)$. Section \ref{surfacNS} introduces the surface analogon of Navier-Stokes equations on $\Gamma(t)$ based on the Boussinesq-Scriven stress tensor, and it is studied from the perspective of applicability as a stationary in space interface condition. Fully coupled, consistent models of flows on $\Omega_-(t) -\Gamma(t) - \Omega_+(t)$ are presented in Section \ref{coupledNS} and the balance laws of momentum and energy are derived. In the remainder of the note some simplified models of stationary fluidic interface, $\Gamma(t)=\Gamma$, are considered and the resulting inconsistencies of balance laws are identified. 

We assume that a reader is familiar with the  tangential calculus.

\section{Two-phase Navier-Stokes equation}\label{NS}

Let $\bu:\Omega_\pm\times [0,T]\to \mathbb{R}^d$ and  $p:\Omega_\pm\times[0,T]\to \mathbb{R}$ be the bulk fluid velocity and the pressure.  The restrictions $\bu^+, p^+$ and $\bu^-, p^-$ are assumed to be $C^1$ and define  bulk flows of separate phases on $\Omega_+$ and $\Omega_-$, correspondingly, which should be coupled by an interface condition on $\Gamma(t)$. We let $$\jump{\bu}:=\bu^--\bu^+$$ denote the jump of bulk velocity; similar notation is used for jumps of other quantities. 

The incompressible Navier-Stokes system for both phases with a prescribed velocity on $\partial\Omega_D$ and presribed stresses on $\partial\Omega_N$ such that $\partial\Omega_+= \partial\Omega_D \cup \partial\Omega_N \cup \Gamma(t)$ reads
\begin{align}
    \rho\, \dot{\bu}&=\div\sigma\qquad \qquad &\text{ on }\Omega_{\pm}(t)\\
    \nabla\cdot{}\bu&=0 &\text{ on }\Omega_{\pm}(t)\\
    \bu&=\bg &\text{ on }\partial\Omega_D\\
    \bt &= \bbf &\text{ on }\partial\Omega_N
\end{align}
where $\dot{\bu}=\bu{}_t + (\nabla \bu) \bu$ is the material derivative, $\sigma=-p\,I+ \mu (\nabla \bu+\nabla^T\bu)$  is the Newtonian stress tensor,  $\rho=\rho^-H+\rho^+(1-H)$ for the indicator function $H$ of $\Omega_-$; the constant densities $\rho_{\pm}$ and the dynamic viscosity $\mu=\mu_{\pm}$ are material parameters of the phases, and $\bt=\bt(\bu)=\sigma\tilde{\bn}$ is the stress vector on $\pa\Omega{}$ oriented by outward pointing $\tilde{\bn}$.
%where $\mu(t):=\mu_+\mathcal{X}_{\Omega_+(t)}+\mu_-\mathcal{X}_{\Omega_-(t)}$ defined as dynamic viscosities.\\
The classical interface conditions on a moving interface:
\begin{itemize}
\item[AI]\textit{Continuous coupling}
\begin{align}
    \jump{\bu}&={0}\qquad \qquad &\text{ on }\Gamma(t)\\
       \jump{\sigma}\bn&={0}&\text{ on }\Gamma(t)\label{ Interface Condition}
   \end{align}
  \item[AII]\textit{Friction slip}
   \begin{align}
    \jump{\bu}\cdot \bn&={0}\qquad \qquad &\text{ on }\Gamma(t)\\
     \bP{\sigma_-\bn}&=-f(\bP\bu^--\bP\bu^+)  &\text{ on }\Gamma(t)\\
  \bP{\sigma_+\bn}&=f(\bP\bu^+-\bP\bu^-)  &\text{ on }\Gamma(t)\\
       \bn\cdot \jump{\sigma}\bn&={0}&\text{ on }\Gamma(t)\label{Jump Interface Condition}
   \end{align}
   \end{itemize}
    where $f$ is the friction coefficient. Note that in both cases the stress vector is continuous across the interface $\Gamma(t)$, $\jump{\sigma}\bn={0}$, which is the well-known in continuum mechanics condition of local conservation of the  momentum flux. 
   
   The geometrical evolution of the interface $\Gamma(t)$ between domains is coupled with the material velocities on $\Omega_+(t)$ and $\Omega_-(t)$:
\begin{align}
    \mathbf{\mathcal{V}}\cdot\bn&=\bu^+\cdot\bn=\bu^-\cdot\bn&\text{ on }\Gamma(t)
   \end{align}
   As we shall see, this coupling leads to a compatibility condition $\int_{\Gamma(t)}\mathcal{V}\cdot{}\bn\,dS=0$.
   \subsection{Mass, momentum and energy  balance laws} 
   We start by showing that continuity of normal component of velocity  across the interface is essential to the mass balance laws. We employ the generalized Reynolds transport theorem to obtain the balance laws of mass $M_-$, momentum $Q_-$ and kinetic energy $E_-$ of the phase that occupies $\Omega_-$:
   \begin{align*}
        &\frac{d}{dt}M_{-}=\frac{d}{dt}\int_{\Omega_-(t)}\!\!\rho \,dV=\int_{\Omega_-(t)}\rho_t  \,dV+\int_{\Gamma(t)}\rho\mathcal{V}\cdot{}\bn=\int_{\Gamma(t)}\rho\bu\cdot{}\bn  \,dS=\int_{\Omega_-(t)}\rho\,\div\bu \,dV=0
        \end{align*}
\begin{align*}
        &\frac{d}{dt}Q_-=\frac{d}{dt}\int_{\Omega_-(t)}\rho \bu\,dV=\int_{\Omega_-(t)}\dot{\rho}\bu+\rho\dot{\bu}+\rho\bu(\nabla\cdot\bu) \,dV=\int_{\Omega_-(t)}\div{}\sigma \,dV =\int_{\Gamma(t)}\sigma\bn \,dS 
        \end{align*}
        \begin{align*}
        &\frac{d}{dt}E_-=\frac{1}{2}\frac{d}{dt}\int_{\Omega_-(t)}\!\!\rho (\bu\cdot\bu)\,dV
        =\int_{\Omega_-(t)}\rho \bu\cdot{}\dot{\bu}\,dV=\int_{\Omega_-(t)} \bu\cdot \div{}\sigma{}\,dV\\&=
       \int_{\Omega_-(t)}\div{}(\sigma{}\bu)\,dV -\int_{\Omega_-(t)}\sigma:\nabla\bu\,dV 
       =\int_{\Gamma(t)} \bn\cdot \sigma{}\bu\,dS -\frac12\int_{\Omega_-(t)} \sigma : (\nabla \bu + \nabla{}\bu^T)dV \\
       &=\int_{\Gamma(t)} \bu\cdot \sigma{}\bn\,dS -\int_{\Omega_-(t)} (-p I+ 2\mu D\bu) : D\bu\,dV=\int_{\Gamma(t)} \bu\cdot \sigma{}\bn\,dS -2\mu\int_{\Omega_-(t)}\| D\bu \|^2\,dV 
      \end{align*}
      where $2 D\bu:=(\nabla \bu + \nabla{}\bu^T)$.
          Similarly, we derive the balance laws for $\Omega_+(t)$ taking into account that $\bn$ points inward to it:
          \begin{align*}
        \frac{d}{dt}M_{+}=-m\,,
        \quad\frac{d}{dt}Q_+=-\int_{\Gamma}\sigma\bn \,dV +r
        \end{align*}
        \begin{align*}
        &\frac{d}{dt}E_+=-2\mu\int_{\Omega_+}\norm{D\bu}^2 dV - \int_{\Gamma}\bu\cdot{}\sigma{}\bn\,dS +R 
\end{align*}

where 
 \begin{align}
m=\int_{\partial{}\Omega}\rho\,\bu\cdot\tilde{\bn}\ \,dS\,,\quad
        &r=  \int_{\pa\Omega_D}\bt{}\,dS+\int_{\pa\Omega_N}{}\bbf{}\,dS
        \end{align}
\begin{align}R=\int_{\pa\Omega}\bu\cdot{}\bt{}\,dS=\int_{\pa\Omega_D}\bg\cdot{}\bt{}\,dS+\int_{\pa\Omega_N}\bu\cdot{}\bbf{}\,dS\end{align} 
is the total outflow rate, the total force and the total mechanical power imposed on the outer boundary $\partial\Omega$.

 Finally, we obtain the total balance laws of bulk phases:
 \begin{align}
\frac{dM}{dt}=\frac{d}{dt}(M_++M_-)=-m\,,\quad
        \frac{dQ}{dt}=\frac{d}{dt}(Q_++Q_-)=\int_{\Gamma}\jump{\sigma{}}\,\bn\,dS +r 
\end{align}     
 \begin{align}
        &\frac{dE}{dt}=\frac{d}{dt}(E_++E_-)=-\int_{\Omega_\pm}2\mu\norm{D \bu}^2 dV + \int_{\Gamma}\jump{\bu\cdot{}\sigma{}}\,\bn\,dS +R
\end{align}
\begin{remark}
The term $ \int_{\Gamma}\jump{\sigma{}}\,\bn\,dS$ is the total interface force which is zero in both interface models AI and AII. The term $ \int_{\Gamma}\jump{\bu\cdot{}\sigma{}}\,\bn\,dS$ is the interface dissipation rate. In case of the continuous coupling model AI this term is zero, while in the case of friction slip model AII it is equal to $-f(\jump{\bu})^2:=-F$.
\end{remark}
\section{Evolving surface Navier-Stokes equation}\label{surfacNS}
We assument the interface $\Gamma$ to be a $C^2$ smooth surface. To formulate equations, we need to further define some notations. There exists a neighborhood $\mathcal{O}(\Gamma)$  of $\Gamma$, we are able to define a well-defined closest point projection as following $\bp:\,\mathcal{O}(\Gamma)\to \Gamma$.  For a surface pressure function $\pi:\, \Gamma \to \mathbb{R}$ or a vector-valued surface velocity field function $\bU:\, \Gamma \to \mathbb{R}^d$  we define $\pi^e=\pi\circ \bp\,:\,\mathcal{O}(\Gamma)\to\mathbb{R}$, $\bU^e=\bU\circ \bp\,:\,\mathcal{O}(\Gamma)\to\mathbb{R}^d$,   extensions of $\pi$ and $\bU$ from $\Gamma$ to its neighborhood $\mathcal{O}(\Gamma)$ along the normal directions.
 The surface gradient and covariant derivatives on $\Gamma$ are then defined as $\nablaG \pi=\bP\nabla \pi^e$ and  $\nabla_\Gamma \bU:= \bP \nabla \bU^e \bP$. The definitions  of surface gradient and covariant derivatives are  independent of a particular smooth extension of $\pi$ and $\bU$ off $\Gamma$.\\
 On $\Gamma$ we consider the surface rate-of-strain tensor %\cite{GurtinMurdoch75}
 given by
\begin{equation} \label{strain}
 \bD_\Gamma(\bU):=\bP \bD\bP= \frac12 \bP (\nabla \bU +\nabla \bU^T)\bP = \frac12(\nabla_\Gamma \bU + \nabla_\Gamma \bU^T).
 \end{equation}
Surface divergence for a vector field $\bg: \Gamma \to \R^3$ and
a tensor field $\bA: \Gamma \to \mathbb{R}^{3\times 3}$ are point-wise defined as:
%\begin{align*}
\[
 \divG \bg  := \tr (\gradG \bg), \qquad
 \divG \bA  := \left( \divG (\be_1^T \bA),\,
               \divG (\be_2^T \bA),\,
               \divG (\be_3^T \bA)\right)^T,
               \]
%\end{align*}
with $\be_i$ the $i$th basis vector in $\R^3$.
 Surface Navier-Stokes equation is a  model of a thin inextensible fluid layer of constant density  $\rho_\Gamma$ and viscosity $\mu_\Gamma$ approximated by a two-dimensional surface $\Gamma(t)$ that moves with velocity $\bU$ subjected to an external force $\bb$. The Boussinesq-Scriven model assumes that stress vectors, which are reactions of the fluidic material, are tangential to the surface.
\begin{itemize}
\item[NS]\textit{Surface Navier-Stokes}
\begin{align}
    \rho_\Gamma \dot{\bU} &= \divG \sigma_\Gamma + \bb&\text{ on }\Gamma(t)\\
    \divG \bU&=0&\text{ on }\Gamma(t)\\
    \sigma_\Gamma &= -\pi\bP +2\mu_\Gamma D_\Gamma(\bU)&\text{ on }\Gamma(t)
\end{align}
\end{itemize}
 It is natural  to link the geometrical evolution of the bulk interface with the velocity $\bU$ of the surface system $\mathcal{NS}_\Gamma(\bU)=\bb$  that occupies it:
\begin{align}
    \mathbf{\mathcal{V}}\cdot\bn&=U_N&\text{ on }\Gamma(t)
   \end{align}

Because of inextensibility of the fluidic layer and the absense of boundary of $\Gamma(t)$ the surface mass $\int_{\Gamma(t)}\rho_\Gamma\bU\,dS$ is conserved according to surface Reynolds transport theorem. To derive the surface momentum and the surface energy balance laws we will need the following 
\begin{lemma}
For a vector $\bg\in$ $({C}^1(\Gamma))^d$ and a matrix $\bG\in ({C}^1(\Gamma))^{d\times{}d}$ such that $\bG=\bP\bG$ we have
\begin{equation}
    \int_{\Gamma} \bg\cdot\divG\bG\bP\,dS=-\int_{\Gamma} \bG:\nabla_{\Gamma}\bg\, dS\label{lemma div sigma}
\end{equation}
\end{lemma}
\begin{proof}
In order to prove the lemma we need 3 steps.

Step 1:\\
For an arbitrary $g\in$ ${C}^1(\Gamma)$ and for all $i\in[1,d]$ consider \begin{align*}\int_{\Gamma}(\nabla_{\Gamma}^Tg)_i \,dS=\int_{\Gamma}(\bP{}\nabla_{\Gamma}^Tg)_i\, dS=\int_{\Gamma}\nabla_{\Gamma}(\bid_i) \cdot{}\nabla_{\Gamma}^Tg\, dS=-\int_{\Gamma}g\Laplace_{\Gamma}\bid_i\,dS=\int_{\Gamma}g\kappa\bn_i\,dS
\end{align*}
where we used the integration by parts for scalar functions on a surface without boundary.
%\begin{align}
 %   \int_{\Gamma}\nabla_{\Gamma}^Tg\,dS=\int_{\Gamma}g\kappa\bn \,dS
%\end{align}

Step 2:\\
For an arbitrary $\blf\in C^1(\Gamma)^{d}$ such that $\blf=\bP\blf$ we replace $g$ with $\blf_i\,g$ and sum for all $i\in[1,d]$:
%\begin{equation}
%    \int_{\Gamma}g (\nabla_{\Gamma}^T \blf_i)_i+\blf_i(\nabla_{\Gamma}^Tg)_i\,dS=\int_{\Gamma}\kappa \blf_i\,g \bn_i\, dS
%    \end{equation}
    \begin{align*}\label{step2}
    \int_{\Gamma}g \divG{} \blf+\blf\cdot(\nabla_{\Gamma}^Tg)\,dS=\int_{\Gamma}\kappa{} g \blf \cdot \bn\, dS
    \end{align*}
    which implies
\begin{align}
    \int_{\Gamma}g\, \div_{\Gamma}\blf\,dS=-\int_{\Gamma}\blf\cdot \nabla_{\Gamma}^Tg dS
\end{align}
Step 3:\\
We let $\bG_i$ denote the $i$-th row of matrix $\bG$ and
\begin{align*}
          &\int_{\Gamma}\bg\cdot\div_{\Gamma}\bG\bP\,dS=\sum_{i=1}^3\int_{\Gamma}g_i\div_{\Gamma}(\be_i^T\bG\bP)^T\,dS=\sum_{i=1}^3\int_{\Gamma}g_i\div_{\Gamma}(\bP\bG^T_i)\,dS
          =-\sum_{i=1}^3\int_{\Gamma}\bP\bG_i^T\cdot \nabla_{\Gamma}^T g_i\,dS\\
          &=-\sum_{i=1}^3\int_{\Gamma} \nabla_{\Gamma}g_i\bP\bG_i^T\,dS=-\sum_{i=1}^3\int_{\Gamma} \nabla{}g_i\bP\bG_i^T\,dS= -\int_{\Gamma}\tr{}(\nabla{}\bg\bP\bG^T) \,dS\\
         &=-\int_{\Gamma}\tr{}(\nabla{}\bg\bP\bG^T \bP )\,dS=-\int_{\Gamma}\tr{}(\bP\nabla{}\bg\bP\bG^T )\,dS=-\int_{\Gamma}\tr{}(\gradG\bg\bG^T )\,dS=-\int_{\Gamma}\bG: \gradG\bg\,dS
\end{align*}
\end{proof}

Now we are ready to show the balance laws for the momentum $Q_\Gamma$ and the kinetic energy $E_\Gamma$ of a fluidic interface with the help of \eqref{lemma div sigma}:
\begin{align}
        &\frac{d}{dt}Q_\Gamma=\frac{d}{dt}\int_{\Gamma(t)}\!\!\rho_\Gamma \bU\,dS=\int_{\Gamma(t)}\left(\rho_\Gamma\dot{\bU}+\dot{\rho}_\Gamma{\bU}+\rho_\Gamma\bU\divG\bU\right) \,dS
        =\int_{\Gamma(t)}\rho_\Gamma \dot{\bU}\,dS\nonumber\\
        &=\int_{\Gamma(t)}  (\divG{}\sigma{}_\Gamma +\bb)\,dS =\int_{\Gamma(t)}   \bb\,dS
        \end{align}
\begin{align}
        &\frac{d}{dt}E_\Gamma=\frac{1}{2}\frac{d}{dt}\int_{\Gamma(t)}\!\!\rho_\Gamma (\bU\cdot\bU)\,dS=\frac{1}{2}\int_{\Gamma(t)}\left(2\rho_\Gamma\bU\cdot{}\dot{\bU}+\rho_\Gamma(\bU\cdot\bU)\divG\bU\right) \,dS\nonumber\\
        &=\int_{\Gamma(t)}\rho_\Gamma \bU\cdot{}\dot{\bU}\,dS=\int_{\Gamma(t)} \bU\cdot (\divG{}\sigma{}_\Gamma +\bb)\,dS=\int_{\Gamma(t)} \bU\cdot\bb\,dS-\int_{\Gamma(t)} \tr (\sigma_{\Gamma}\nabla_{\Gamma}\bU)\, dS \label{lemma 3.2 application} \nonumber\\
                 &=\int_{\Gamma(t)} \bU\cdot \bb\,dS -2\mu_\Gamma\int_{{\Gamma(t)}} \| D_\Gamma{}\bU\|^2  \,dS
\end{align}

       where we clearly see the total external force, the total external work in the total viscous dissipation of the fluidic interface.
     % We will denote the energy needed to evolve the surface against the stresses as  
     % \begin{align}
  %        L_\Gamma=\int_{\Gamma(t)}U_N %(\sigma_\Gamma :\bH)\,dS
  %    \end{align}
  %    which, up to some viscosity terms, equals
  %   \begin{align}
   %      -\int_{\Gamma(t)}U_N \pi \kappa\,dS=\int_{\Gamma(t)} \pi \divG\bU_T\,dS = -\int_{\Gamma(t)}  \bU_T\cdot{}\gradG^T\pi\,dS 
  %   \end{align}

   \subsection{Surface Navier-Stokes as a stationary interface}
Here we would like to understand under which conditions the Bousinesque-Scriven model can be used to model a stationary in space fluidic interface. We start by splitting the surface Navier-Stokes system into normal and tangential parts (see e.g. in \cite{jankuhn2017incompressible}):
\begin{align}
        \rho_{\Gamma}\bP\dot{\bU}_T&=-\nabla_{\Gamma}^T\pi+2\mu_{\Gamma}P\divG D_\Gamma(\bU)+b_T-\rho_{\Gamma} U_N\dot{\bn}&\text{ on }\Gamma(t)\\
        \rho_{\Gamma}\dot{U}_N&=2\mu_{\Gamma}\bn\cdot\divG D_\Gamma(\bU)+\pi \kappa+b_N+\rho_{\Gamma}\dot{\bn}\cdot \bU_T &\text{ on }\Gamma(t)\\
         \divG \bU &=0 &\text{ on }\Gamma(t)
   \end{align}
   Let us consider surface Euler equations assuming $\mu_\Gamma=0$ and demonstrate how momentum and energy split in the normal and tangential directions. Taking inner product of the first and the second equations with $\bU_T$ and  $U_N$  correspondingly we arrive at the energy law:
     \begin{align*}
         &\int_{\Gamma(t)}\rho_\Gamma\left(\dot{\bU}_T\cdot \bU_T + \dot{U}_N{}U_N\right) \,dS=\int_{\Gamma(t)}(\bU\cdot\bb +U_N\pi\kappa - \bU_T \cdot{} \gradG^T\pi )\,dS\\
         &=\int_{\Gamma(t)}\bU\cdot\bb \,dS+\int_{\Gamma(t)}\pi(U_N\kappa +  \divG\bU_T)  \,dS=\int_{\Gamma(t)}\bU\cdot\bb \,dS
     \end{align*}
     Since 
     \begin{align*}
        \dot{\bU}\cdot\bU&=(\dot\bU_T + \dot{U}_N \bn + U_N \dot{\bn})\cdot(\bU_T+U_N\bn)=\dot\bU_T\cdot\bU_T + \dot{U}_N \cdot{}U_N \\ 
        +& U_N{}(\dot\bU_T\cdot\bn + \bU_T\cdot\dot\bn )=\dot\bU_T\cdot\bU_T + \dot{U}_N \cdot{}U_N  
        \end{align*}
        we obtain the splitting of the surface energy balance law
        \begin{align*}\frac{dE_\Gamma}{dt}=\frac12\frac{d}{dt}\int_{\Gamma(t)}\rho_\Gamma{\bU}^2 \,dS=\frac12\frac{d}{dt}\int_{\Gamma(t)}\rho_\Gamma\left({\bU}^2_T + U_N^2\right) \,dS=\int_{\Gamma(t)}\bU\cdot\bb \,dS
        \end{align*}

Now we integrate the first and the second equations assuming $\mu_\Gamma=0$ to obtain the directional split of the momentum law:
     \begin{align*}
         &\int_{\Gamma(t)}\rho_\Gamma\left( \bP\dot{\bU}_T + \dot{U}_N\bn\right) \,dS=\int_{\Gamma(t)}(\bb +\pi\kappa\bn - \gradG^T\pi  -\rho_{\Gamma} U_N\dot{\bn} + \rho_{\Gamma}(\dot{\bn}\cdot \bU_T) \bn)\,dS\\
        &=\int_{\Gamma(t)}\bb \,dS-\int_{\Gamma(t)}\rho_{\Gamma} (U_N\dot{\bn} + ({\bn}\cdot \dot{\bU}_T) \bn) \,dS%=\int_{\Gamma(t)}\bb \,dS-C_\Gamma
     \end{align*}
    where $\dot{\bn}\cdot{\bU}_T =-{\bn}\cdot \dot{\bU}_T $ is used. Noticing $\dot{\bP}\bU_T+\bP\dot{\bU}_T=\dot{\bU}_T$ and using
     \begin{align}
        \dot{\bU}&=\dot\bU_T + \dot{U}_N \bn + U_N \dot{\bn}=\bP\dot\bU_T + \dot{\bP}\bU_T + \dot{U}_N \bn + U_N \dot{\bn}
        \end{align}
        we conclude that a part of the momentum balance is the change of velocity on the surface and another part is due to the geometrical evolution:
        \begin{align}\frac{dQ_\Gamma}{dt}=\int_{\Gamma(t)}\rho_\Gamma\left( \bP\dot{\bU}_T + \dot{U}_N\bn\right) \,dS +\int_{\Gamma(t)}\rho_{\Gamma} (U_N\dot{\bn} + \dot{\bP}\bU_T)\,dS=\int_{\Gamma(t)}\bb \,dS
        \end{align}
     It is easy to see that the last statement also holds in case of $\mu_\Gamma\neq0$. 

Now we derive a model of a fluidic interface that corresponds to a stationary surface in space, i.e. $U_N=0$. Unfortunately, a naive insertion of this condition into the surface Navier-Stokes does not lead to a consistent system. Indeed, set $U_N=0$ and have
   \begin{align}
        \rho_{\Gamma}\bP\dot{\bU}_T&=-\nabla_{\Gamma}^T\pi+2\mu_{\Gamma}\bP\divG D_\Gamma(\bU_T)+\bb_T&\text{ on }\Gamma(t)\\
        0&=-2\mu \tr(H\nabla_\Gamma \bU_T) +\pi \kappa+b_N+\rho_{\Gamma}\bU_T\cdot \bH\bU_T &\text{ on }\Gamma(t)\\
         \divG \bU_T &=0 &\text{ on }\Gamma(t)
   \end{align}
   where we used the identity $\dot{\bn}=\bH \bU_T - \gradG^TU_N$.
   It is clear that the first and the third equation define a geometric PDE (with the  covariant material derivative $\bP\dot{\bU}_T$) on a stationary surface with a solution $\bU_T$ and $\pi{}$. However, the second equation will be satisfied if only there is an external force $b_N$ that balances other normal forces:
   \begin{align}
       b_N=-\pi \kappa +(2\mu \tr(H\nabla_\Gamma \bU_T) - \rho_{\Gamma}\bU_T\cdot \bH\bU_T)
   \end{align}
   We are assuming that this normal force is always applied to keep the position of the surface, and is composed of a Laplace force, $-\pi\kappa$, and a normal force $N(\bU_T)$ that depends on the tangential motion:
   \begin{align}
       N(\bU_T)=2\mu \tr(H\nabla_\Gamma \bU_T) - \rho_{\Gamma}\bU_T\cdot \bH\bU_T
   \end{align}  
   where the first term is a normal viscous traction caused by the curvatures and the last term is so-called centripetal force. Essentially, all three normal forces need to be balanced by the external $b_N$ to keep the position of the surface.
   
   The tangential system, which we denote as $\mathcal{NS}^*(\bU_T)=\bb_T$, is coupled with the external to the surface tangent force as follows:
   \begin{align}
        \rho_{\Gamma}\bP\dot{\bU}_T&=-\nabla_{\Gamma}^T\pi+2\mu_{\Gamma}P\divG D_\Gamma(\bU_T)+\bb_T&\text{ on }\Gamma(t)\\
         \divG \bU_T &=0 &\text{ on }\Gamma(t)
   \end{align}
 Let us summarize the suggested model of a stationary fluidic interface:
   \begin{itemize}
\item[NS*]\textit{Stationary surface Navier-Stokes}
\begin{align}
   U_N&=0\\
   \mathcal{NS}^*(\bU_T)&=\bb_T\\
    b_N&=-\pi \kappa + N(\bU_T)
   \end{align}
   \end{itemize}
   
      The momentum and the energy balance laws of a stationary fluidic interface can be derived similarly to the general case. We consider the equation $\mathcal{NS}^*(\bU_T)=\bb_T$ solely  and compute the following:
   \begin{align}
  &\int_{\Gamma(t)}\rho_\Gamma \dot{\bU}_T  \,dS=\frac{\delta{}Q^p_\Gamma}{\delta{}t}+\frac{\delta{}Q^c_\Gamma}{\delta{}t}= \int_{\Gamma(t)}\rho_\Gamma \bP\dot{\bU}_T  \,dS +\int_{\Gamma(t)}\rho_{\Gamma} \dot{\bP}\bU_T\,dS\\
   &=\int_{\Gamma(t)}(-\nabla_{\Gamma}^T\pi+2\mu_{\Gamma}\bP\divG D_\Gamma(\bU_T)+\bb_T)\,dS
   +\int_{\Gamma(t)}\rho_{\Gamma}(\dot{\bU}_T\cdot\bn)\bn\, dS\\
   &=\int_{\Gamma(t)}\bb_T\,dS+\int_{\Gamma(t)}(-\pi\kappa\bn+2\mu_{\Gamma}\bP\divG D_\Gamma(\bU_T)-\rho_{\Gamma}( \bU_T \cdot \bH \bU_T)\bn\,dS\\
    &=\int_{\Gamma(t)}\bb_T\,dS+\int_{\Gamma(t)}(-\pi\kappa\bn+2\mu \tr(H\nabla_\Gamma \bU_T)-\rho_{\Gamma}( \bU_T \cdot \bH \bU_T)\bn\,dS\\
      &=\int_{\Gamma(t)}\bb_T\,dS+\int_{\Gamma(t)}(-\pi \kappa + N(\bU_T))\bn\,dS\\
       &=\frac{\delta{}Q^*_\Gamma}{\delta{}t} +\int_{\Gamma(t)}\rho_{\Gamma} (\bU_T \cdot \bH \bU_T)\bn\,dS=\frac{d{}Q^*_\Gamma}{dt}
        \end{align}
        
        Similarly,
         \begin{align}
 &\frac{d}{dt}E^*_\Gamma=\frac12\frac{d}{dt}\int_{\Gamma(t)}\rho_\Gamma \bU_T^2 \,dS=\int_{\Gamma(t)}\rho_\Gamma \bU_T\cdot\dot{\bU}_T  \,dS= \int_{\Gamma(t)}\rho_\Gamma \bU_T\cdot\bP\dot{\bU}_T  \,dS\\
 &=\int_{\Gamma(t)}\bU_T\cdot \bb_T \,dS+\int_{\Gamma(t)}(-\bU_T\cdot\nabla_{\Gamma}^T\pi+2\mu_{\Gamma}\bU_T\cdot\bP\divG D_\Gamma(\bU_T))\,dS\\
 &=\int_{\Gamma(t)}\bU_T\cdot \bb_T \,dS +\int_{\Gamma}\pi\, \div_{\Gamma}\bU_T\,dS -2\mu_{\Gamma}\int_{\Gamma} D_\Gamma(\bU_T):\nabla_{\Gamma}\bU_T\, dS\\
  \end{align}
        
        Finally, the following a priori estimates hold:
        \begin{align}
        &\frac{d{}Q^*_\Gamma}{dt}=\int_{\Gamma(t)}\bb_T \,dS+\int_{\Gamma(t)}(-\pi \kappa + N(\bU_T))\bn\,dS\\
        &\frac{dE^*_\Gamma}{dt} =\int_{\Gamma(t)}   \bU_T\cdot\bb_T\,dS -2\mu_\Gamma\int_{{\Gamma(t)}} \| D_\Gamma{}\bU_T \|^2 \,dS
        \end{align}
  
  \begin{remark} 
  The suggested above momentum and energy laws are due to the structure of $\mathcal{NS}^*(\bU_T)=\bb_T$ equations. If one takes into account $U_N=0$ and $b_N=-\pi \kappa + N(\bU_T)$ then the non-split momentum law of surface Navier-Stokes can be recovered.
 
  \end{remark}
\section{Two-phase flow with a fluidic interface}\label{coupledNS}
Our goal is to couple the surface flow and the bulk flow such that momentum and energy balance laws are kept valid a priori, similar to the mass balance law being correct automatically once the geometrical compatibility condition is fulfilled. We can couple these systems kinematicaly via velocities,  dynamically via forces or any combination of them so the total number of interface conditions is consistent with the number of conditions on an interface between bulk flows, which is 6. 
The kinematical coupling can be continuous and discontinuous, however the local conservation of the mass guarantees that the normal bulk velocity should be continuous across the interface.
Assuming there are no purely external forces like gravity, the dynamical coupling should be performed through an unknown force $\bb$ that, from one hand, enters the surface flow momentum equation, and from another hand balances the jump of stress vector across the interface to guarantee local conservation of the momentum flux: $$\jump{\sigma}\bn+\bb=0$$
Since the surface Navier-Stokes system includes intrinsic to the interface variables $\bU$ and $\bb$ we actually need $6+3+3=12$ condition to define a fluidic interface. This coupling technique results in the two models BI and BII.

\begin{itemize}
   \item[BI] \textit{Continuous coupling with fluidic interface}
\begin{align}
 \jump{\bu}&=0\qquad \qquad &\text{ on }\Gamma(t)\\
       \jump{\sigma}\bn+\bb&=0&\text{ on }\Gamma(t)\\
       \bU &= \bu{}\hfill{}&\text{ on }\Gamma(t)\\
         \mathcal{NS}_\Gamma(\bU)&=\bb&\text{ on }\Gamma(t)
   \end{align}

 %\item[D] \textit{Perfect slip on fluidic interface}
 %\begin{align}
 %\jump{\bu }\cdot \bn&=0&\text{ on }\Gamma(t)\\
 %      \bP{\sigma_-\bn}&=0&\text{ on }\Gamma(t)\\
  %     \bP{\sigma_+\bn}&=0&\text{ on }\Gamma(t)\\
   %      \jump{\sigma}\bn+\bb&=0 &\text{ on }\Gamma(t)
    %     \\
     %      U_N &= u_N\hfill{}&\text{ on }\Gamma(t)\\
      %   \mathcal{NS}_\Gamma(\bU)&=\bb&\text{ on }\Gamma(t)
   %\end{align}
 \end{itemize}

    For the interface above $\jump{\sigma{}}\,\bn=-\bb$ and $\jump{\bu\cdot{}\sigma{}}\,\bn=-\bU\cdot{}\bb$ and the total bulk momentum  and energy  can be expressed as follows:
     \begin{align}
        &\frac{dQ}{dt}=-\int_{\Gamma}\bb\,dS +r \label{Total energy}
\end{align}
    \begin{align}
        &\frac{dE}{dt}=-\int_{\Omega_\pm}2\mu\norm{D \bu}^2 dV - \int_{\Gamma(t)}\bU\cdot{}\bb\,dS +R 
\end{align}
or, with the help of the balance of surface momentum and energy:
\begin{align}
        &\frac{d}{dt}(Q+Q_\Gamma)=r \label{super_Total energy}
\end{align}
  \begin{align}
        &\frac{d}{dt}(E+E_\Gamma)=-\int_{\Omega_\pm(t)}2\mu\norm{D \bu}^2 dV -2\mu_\Gamma\int_{\Gamma(t)}\norm{D_\Gamma \bU}^2 dS  +R 
\end{align}

\begin{itemize}
   \item[BII] \textit{Friction slip on fluidic interface}
 \begin{align}
 \jump{\bu }\cdot \bn&=0&\text{ on }\Gamma(t)\\
   u_N &= U_N\hfill{}&\text{ on }\Gamma(t)\\
      \bP{\sigma_-\bn}&=-f_-(\bP\bu^--\bU_T)  &\text{ on }\Gamma(t)\\
  \bP{\sigma_+\bn}&=f_+(\bP\bu^+-\bU_T)  &\text{ on }\Gamma(t)\\
         \jump{\sigma}\bn+\bb&=0 &\text{ on }\Gamma(t)
         \\
         \mathcal{NS}_\Gamma(\bU)&=\bb&\text{ on }\Gamma(t)
   \end{align}
   \end{itemize}
  
   For the  interface above $\jump{\sigma{}}\,\bn=-\bb$ and 
   \begin{align*}
       \jump{\bu\cdot{}\sigma{}}\,\bn=-\bU\cdot{}\bb -f_-(\bP\bu^--\bU_T)^2-f_+(\bP\bu^+-\bU_T)^2
   \end{align*} where we extracted the term for the friction dissipation
    \begin{align}
        F_\pm= \int_\Gamma f_-(\bP\bu^--\bU_T)^2+f_+(\bP\bu^+-\bU_T)^2 \,dS
    \end{align}
    and the total bulk momentum and  energy  can be expressed as follows:
      \begin{align}
        &\frac{dQ}{dt}=-\int_{\Gamma}\bb\,dS +r \label{neTotal energy}
\end{align}
    \begin{align}
        &\frac{dE}{dt}=-\int_{\Omega_\pm}2\mu\norm{D \bu}^2 dV - \int_{\Gamma(t)}\bU\cdot{}\bb\,dS -F_\pm +R 
\end{align}
or, with the help of the balance of surface momentum and energy:
\begin{align}
        &\frac{d}{dt}(Q+Q_\Gamma)=r \label{nesuperTotal energy}
\end{align}
  \begin{align}
        &\frac{d}{dt}(E+E_\Gamma)=-\int_{\Omega_\pm(t)}2\mu\norm{D \bu}^2 dV -2\mu_\Gamma\int_{\Gamma(t)}\norm{D_\Gamma \bU}^2 dS -F_\pm+R 
\end{align}

Both presented models have correct total mass, momentum and energy balance laws.
\subsection{Stationary fluidic interfaces}
Here we would like to consider the possibility of a coupling of the stationary fluidic interface with bulk flows for modeling purposes. Essentially, we replace the general fluidic interface model
$
    \mathcal{NS}(\bU)=\bb
$
with its stationary counterpart derived previously:
\begin{align}
 U_N&=0\\
   \mathcal{NS}^*(\bU_T)&=\bb_T\\
    b_N&=-\pi \kappa+N(\bU_T)
\end{align}
    Notice that thus we replace 3 conditions with 4 conditions and one condition should be relaxed. However, we choose not to relax the continuity of momentum flux in order to keep the total momentum balance law intact in the suggested models below.
\begin{itemize}
 \item[SI] \textit{Continuous coupling on stationary fluidic interface}
\begin{align}
 \jump{\bu}&=0\qquad \qquad &\text{ on }\Gamma(t)\\
       \jump{\sigma}\bn+\bb&=0&\text{ on }\Gamma(t)\\
       \bU &= \bu{}\hfill{}&\text{ on }\Gamma(t)\\
       U_N&=0&\text{ on }\Gamma(t)\\
   \mathcal{NS}^*(\bU_T)&=\bb_T&\text{ on }\Gamma(t)%\\b_N&=-\pi \kappa+N(\bU_T)&\text{ on }\Gamma(t)
   \end{align}
   
Here we choose to relax  the normal force equation, $b_N=-\pi \kappa+N(\bU_T)$, of the surface Navier-Stokes equation. This means that $b_N$ will be computed from the momentum flux equation, $\jump{\sigma}\bn+\bb=0$, after the whole coupled system is solved, and it may be different from the value of the relaxed condition. Indeed, the total bulk momentum and  energy laws can be expressed as follows:
\begin{align}
        &\frac{d}{dt}(Q+Q_\Gamma)=\frac{d}{dt}(Q+Q^*_\Gamma)=\int_{\Gamma}(b_N + \pi \kappa-N(\bU_T)) \bn\,dS +r \label{ha_Total energy}
\end{align}
  \begin{align}
        &\frac{d}{dt}(E+E_\Gamma)=\frac{d}{dt}(E+E^*_\Gamma)=-\int_{\Omega_\pm(t)}2\mu\norm{D \bu}^2 dV -2\mu_\Gamma\int_{\Gamma(t)}\norm{D_\Gamma \bU}^2 dS  +R 
\end{align}
  
   \item[qSI] \textit{Continuous coupling on quasi stationary fluidic interface}
\begin{align}
 \jump{\bu}&=0\qquad \qquad &\text{ on }\Gamma(t)\\
       \jump{\sigma}\bn+\bb&=0&\text{ on }\Gamma(t)\\
       \bU &= \bu{}\hfill{}&\text{ on }\Gamma(t)\\
       %U_N&=0\\
   \mathcal{NS}^*(\bU_T)&=\bb_T&\text{ on }\Gamma(t)\\
    b_N&=-\pi \kappa+N(\bU_T)&\text{ on }\Gamma(t)
   \end{align}
   
   Here we choose to relax the condition of a stationary surface, $U_N = 0$. This means that while we use the equilibrium force  $b_N=-\pi \kappa+N(\bU_T)$ to balance the bulk stress jump, we cannot expect surface to be at the same position.
  And the total bulk momentum and  energy  can be expressed as follows:
\begin{align}
        &\frac{d}{dt}(Q+Q_\Gamma) - \int_{\Gamma(t)}\rho_\Gamma (\dot{U}_N \bn +{U}_N \dot{\bn})   \,dS =\frac{d}{dt}(Q+Q^*_\Gamma)=r \label{ad_Total energy}
\end{align}
  \begin{align}
        &\frac{d}{dt}(E+E_\Gamma)-\int_{\Gamma(t)}\rho_\Gamma{}U_N\dot{U}_N\,dS=\frac{d}{dt}(E+E^*_\Gamma)\\
        &=-\int_{\Gamma{}(t)}{U_N (-\pi \kappa+N(\bU_T))}\,dS-\int_{\Omega_\pm(t)}2\mu\norm{D \bu}^2 dV -2\mu_\Gamma\int_{\Gamma(t)}\norm{D_\Gamma \bU_T}^2 dS +R 
\end{align} 

%\item[SII]\textit{No-slip condition on interface}
%\begin{align}
 %\jump{\bu}&=0\qquad \qquad &\text{ on }\Gamma(t)\\
  %     \jump{\sigma}\bn&=0&\text{ on }\Gamma(t)\\
   %    \bU &= \bu{}\hfill{}&\text{ on }\Gamma(t)\\
    %     \mathcal{NS}_\Gamma(\bU)&=\bb_T^e+\jump{\bP\sigma\bn}&\text{ on }\Gamma(t)
   %\end{align}
\item[SII]\textit{Friction slip on stationary fluidic interface with external force}
\begin{align}
  \jump{\bu }\cdot \bn&=0&\text{ on }\Gamma(t)\\
   u_N &= U_N\hfill{}&\text{ on }\Gamma(t)\\
      \bP{\sigma_-\bn}&=-f_-(\bP\bu^--\bU_T)  &\text{ on }\Gamma(t)\\
  \bP{\sigma_+\bn}&=f_+(\bP\bu^+-\bU_T)  &\text{ on }\Gamma(t)\\
         \jump{\sigma}\bn+\bb&=0 &\text{ on }\Gamma(t)\\
%U_N &= 0\hfill{}&\text{ on }\Gamma(t)\\
   U_N&=0&\text{ on }\Gamma(t)\\
   \mathcal{NS}^*(\bU_T)&=\bb_T + \bb^e_T&\text{ on }\Gamma(t)
   \end{align}
Here we choose to relax  the normal force equation, $b_N=-\pi \kappa+N(\bU_T)$, of the surface Navier-Stokes equation. This means that $b_N$ will be computed from the momentum flux equation, $\jump{\sigma}\bn+\bb=0$, after the whole coupled system is solved, and it may be different from the value of the relaxed condition.
And the total bulk momentum can be expressed as follows:
\begin{align}
        &\frac{d}{dt}(Q+Q_\Gamma)=\int_{\Gamma}\bb_N\,dS +r \label{bu_Total energy}
\end{align}
   The total energy balance law for the interface above can be derived similarly to the BII interface:
  \begin{align}
    \frac{dE}{dt}+\frac{dE_{\Gamma}}{dt}&=
    -\int_{\Omega_\pm}2\mu\norm{D \bu}^2 dV-2\mu_\Gamma\int_{{\Gamma(t)}} \| D_\Gamma{}\bU_T\|^2\,dS+F^e -F_\pm+R
\end{align}
    where  
\begin{align}
        F^e= \int_{\Gamma}\bU_T\cdot\bb^e_T\,dS
    \end{align}

   \item[qSII]\textit{Friction slip on quisi stationary fluidic interface with external force}
\begin{align}
  \jump{\bu }\cdot \bn&=0&\text{ on }\Gamma(t)\\
   u_N &= U_N\hfill{}&\text{ on }\Gamma(t)\\
      \bP{\sigma_-\bn}&=-f_-(\bP\bu^--\bU_T)  &\text{ on }\Gamma(t)\\
  \bP{\sigma_+\bn}&=f_+(\bP\bu^+-\bU_T)  &\text{ on }\Gamma(t)\\
         \jump{\sigma}\bn+\bb&=0 &\text{ on }\Gamma(t)\\
%U_N &= 0\hfill{}&\text{ on }\Gamma(t)\\
   \mathcal{NS}^*(\bU_T)&=\bb_T+\bb^e_T&\text{ on }\Gamma(t)\\
   b_N&=-\pi \kappa+N(\bU_T)&\text{ on }\Gamma(t)
   \end{align}
   
   \end{itemize}
   
   Here we choose to relax $U_N = 0$, the condition of a purely tangential interface motion. As a consequence, while we use the equilibrium force  $b_N=-\pi \kappa+N(\bU_T)$ to balance the bulk stress jump, we cannot expect surface to be at the same position. Since the $U_N\neq0$ it may be more consistent to keep the $\kappa U_N$ term in the incompressibility condition of  $\mathcal{NS}^*(\bU_T)=\bb_T+\bb^e$.
   And the total bulk momentum can be expressed as follows:
\begin{align}
        &\frac{d}{dt}(Q+Q_\Gamma)=r \label{bulk_Total energy}
\end{align}
 
   The energy balance law for the interface above can be derived similarly to the BII interface:
  \begin{align}
    \frac{dE}{dt}+\frac{dE^*_{\Gamma}}{dt}&=-\int_{\Gamma}U_N b_N \,dS
    -\int_{\Omega_\pm}2\mu\norm{D \bu}^2 dV-2\mu_\Gamma\int_{{\Gamma(t)}} \| D_\Gamma{}\bU_T\|^2\,dS+F^e -F_\pm+R
\end{align}

\section{Discussion}
 In this note we have studied balance laws of different models of bulk-surface-bulk flows which is helpful for the future research.
 
 In the previous section we presented a kinematically coupled model BI and a dynamically coupled model BII of bulk-surface-bulk flows. Both models have physically correct mass, energy, momentum structures. These structures are considered as references in comparison with the structures of simplified models SI, qSI, SII, qSII. 

\bibliographystyle{plain}
\bibliography{literatur}

\end{document}